\documentclass[11pt]{article}

\usepackage{fullpage,amsthm,amsmath,algorithm,algorithmic,multirow,tabularx,paralist,color}
\usepackage[colorlinks=true,pdfpagemode=none,linkcolor=blue,citecolor=blue,pdfstartview=FitH]{hyperref}

\newcommand{\INPUT}{\item[{\bf Input:}]}
\newcommand{\OUTPUT}{\item[{\bf Output:}]}

\newtheorem{theorem}{Theorem}[section]
\newtheorem{lemma}[theorem]{Lemma}

\newtheorem{corollary}[theorem]{Corollary}

\newtheorem{observation}[theorem]{Observation}
\newtheorem{clm}[theorem]{Claim}

\newcommand{\E}[1]{{\bf{E}}\left[#1\right]}
\newcommand{\EE}[2]{{\bf{E}}_{#1}\left[#2\right]}
\renewcommand{\P}[1]{{\bf{P}}\left[#1\right]}
\newcommand{\PP}[2]{{\bf{P}}_{#1}\left[#2\right]}
\renewenvironment{proof}{\noindent{\em Proof.}~}{$\hfill \qed$\\}
\newenvironment{proofof}[1]{{\noindent \em Proof of #1.  }}{\hfill\qed}

\def\b1{{\bf 1}}

\def\ZZ{{\mathbb Z}}
\def\cM{{\cal M}}
\def\cI{{\cal I}}

\def\cE{{\cal E}}
\def\cP{{\cal P}}

\def\eps {\epsilon}
\def\max{{\rm{max}}}

\def \opt{{\rm{OPT}}}
\def \alg{{\rm{ALG}}}
\def \weight{\omega}
\def \L{L}
\begin{document}

\title{On Variants of the Matroid Secretary Problem}
\author{
Shayan Oveis Gharan\thanks{Stanford University, Stanford, CA; {\tt shayan@stanford.edu};
this work was done while the author was at IBM Almaden Research Center, San Jose, CA.}
\and Jan Vondr\'ak\thanks{IBM Almaden Research Center, San Jose, CA; {\tt jvondrak@us.ibm.com}}
}

\date{}

\maketitle
\begin{abstract}
We present a number of positive and negative results for variants of the matroid secretary problem.
Most notably, we design a constant-factor competitive algorithm for the ``random assignment'' model where the weights are assigned randomly to the elements of a matroid, and then the elements arrive on-line in an adversarial order (extending a result of Soto \cite{Soto11}).
This is under the assumption that the matroid is known in advance. If the matroid is unknown in advance, we present
an $O(\log r \log n)$-approximation, and prove that a better than $O(\log n / \log \log n)$ approximation is impossible. This resolves
an open question posed by Babaioff et al. \cite{BIK07}. 

As a natural special case, we also consider the classical secretary problem where the number of candidates $n$ is unknown in advance.
If $n$ is chosen by an adversary from $\{1,\ldots,N\}$, we provide a nearly tight answer, by providing an algorithm that chooses the best candidate with probability at least $1/(H_{N-1}+1)$ and prove that a probability better than $1/H_N$
cannot be achieved (where $H_N$ is the $N$-th harmonic number).
\end{abstract}

\section{Introduction}

The secretary problem is a classical problem in probability theory, with obscure origins in the 1950's and early 60's (\cite{Gardner60,Lindley61,Dynkin63}; see also \cite{Ferguson89}).
The goal in this problem is to select the best candidate out of a sequence revealed one-by-one, where the ranking is uniformly random.
A classical solution finds the best candidate with probability at least $1/e$ \cite{Ferguson89}.
Over the years a number of variants have been studied, starting with \cite{GM66} where multiple choices and various measures of success were considered for the first time.

Recent interest in variants of the secretary problem has been motivated by applications in on-line mechanism design \cite{HKP04,Kleinberg05,BIK07},
where items are being sold to agents arriving on-line, and there are certain constraints on which agents can be simultaneously satisfied.
Equivalently, one can consider a setting where we want to hire several candidates under certain constraints.
Babaioff, Immorlica and Kleinberg \cite{BIK07} formalized the {\em matroid secretary problem} and presented constant-factor competitive algorithms for several interesting cases. The general problem formulated in \cite{BIK07} is the following.

\paragraph{Matroid secretary problem.}
Given a matroid $\cM = (E, \cI)$ with non-negative weights assigned to $E$;
the only information known up-front is the number of elements $n:=|E|$.
The elements of $E$ arrive in a random order, with their weights revealed as they arrive.
When an element arrives, it can be selected or rejected.
The selected elements must always form an independent set in $\cM$, and
a rejected element cannot be considered again.
The goal is to maximize the expected weight of the selected elements.

\


Additional variants of the matroid secretary problem have been proposed and studied, depending on how the input ordering is generated,
how the weights are assigned and what is known in advance. In all variants, elements with their weights arrive in an on-line fashion
and an algorithm must decide irrevocably whether to accept or reject an element once it has arrived.
We attempt to bring some order to the multitude of models and
we classify the various proposed variants as follows.

\medskip

{\em Ordering of matroid elements on the input:}
\begin{compactitem}
\item AO = Adversarial Order: the ordering of elements of the matroid on the input is chosen by an adversary.
\item RO = Random Order: the elements of the matroid arrive in a random order.
\end{compactitem}

\medskip

{\em Assignment of weights:}
\begin{compactitem}
\item AA = Adversarial Assignment: weights are assigned to elements of the matroid by an adversary.
\item RA = Random Assignment: the weights are assigned to elements by a random permutation of an adversarial set of weights
(independent of the input order, if that is also random).
\end{compactitem}

\medskip

{\em Prior information:}
\begin{compactitem}
\item MK = Matroid Known: the matroid is known beforehand (by means of an independence oracle).
\item MN = Matroid - $n$ known: the matroid is unknown but the cardinality of the ground set is known
beforehand.
\item MU = Matroid - Unknown: nothing about the matroid is known in advance; only subsets of the elements
that arrived already can be queried for independence.
\end{compactitem}

\medskip

For example, the original variant of the matroid secretary problem \cite{BIK07},
where the only information known beforehand is the total number of elements, can be described as RO-AA-MN in this classification.
We view this as the primary variant of the matroid secretary problem.

We also consider variants of the classical secretary problem; here, only 1 element should be chosen
and the goal is to maximize the probability of  selecting the best element.

\medskip

{\em Classical secretary problems:}
\begin{compactitem}
\item CK = Classical - Known $n$: the classical secretary problem where the number of elements in known in advance.
\item CN = Classical - known upper bound $N$: the classical secretary problem where the number of elements
is chosen adversarially from $\{1,\ldots,N\}$, and $N$ is known in advance.
\item CU = Classical - Unknown $n$: the classical secretary problem where no information on the number of elements
is known in advance.
\end{compactitem}

\medskip

\noindent
Since the independent sets of the underlying matroid in this model are independent of the particular labeling of the ground set (i.e., RO-AA-CK, AO-RA-CK and RA-RO-CK models are equivalent), we just use the weight assignment function to characterize different variants of this model.
The classical variant of the secretary problem which allows a $1/e$-approximation would be described as RA-CK.
The variant where the number of elements $n$ is not known in advance is very natural --- and has been considered
under different stochastic models where $n$ is drawn from a particular distribution \cite{Stewart80,ABT82} ---
but the worst-case scenario does not seem to have received attention. We denote this model RA-CU,
or RA-CN if an upper bound on the number of candidates is given.
In the model where the input ordering of weights is adversarial (AA-CK), it is easy to see that no algorithm
achieves probability better than $1/n$ \cite{BS09}.
We remark that variants of the secretary problem with other objective functions have been also proposed,
such as discounted profits \cite{BDGIT09}, and submodular objective functions \cite{BHZ10,GRST10}.
We do not discuss these variants here.

\subsection{Recent related work}


The primary variant of matroid secretary problem (RO-AA-MN model)
was introduced in \cite{BIK07}.
In the following, let $n$ denote the total number of elements and $r$ the rank of the matroid.
An $O(\log r)$-approximation for the RO-AA-MN model was given in \cite{BIK07}.
It was also conjectured that a constant-factor approximation should exist for this problem and this question is still open.
Very recently, Chakraborty and Lachish \cite{CL12} improved  \cite{BIK07} by giving an $O(\sqrt{\log r})$-approximation algorithm.
Constant-factor approximations were given in \cite{BIK07}
for some special cases such as partition matroids and graphic matroids with a given explicit representation.
Further, constant-factor approximations were given for transversal matroids \cite{DP08,KP09}
and laminar matroids \cite{IW11}.
However, even for graphic matroids in the RO-AA-MK model when the graphic matroid is given by an oracle,
no constant factor is known.

Babaioff et al. in \cite{BIK07} also posed as an open problem whether there is a constant-factor approximation algorithm for the following two models:
Assume that a set of $n$ numerical values are assigned to the matroid elements using a random one-to-one correspondence but that the elements are presented in an adversarial order (AO-RA in our notation).
Or, assume that both the assignment of values and the ordering
of the elements in the input are random (RO-RA in our notation).
The issue of whether the matroid is known beforehand is left somewhat ambiguous in \cite{BIK07}.

In a recent work \cite{Soto11}, Jos\'e Soto partially answered the second question, by designing a constant-factor approximation algorithm in the RO-RA-MK model: An adversary chooses a list of non-negative weights,
which are then assigned to the elements using a random permutation, which is independent of the random order
at which the elements are revealed. The matroid is known in advance here.


\subsection{Our results}

\paragraph{Matroid secretary.}
We resolve the question from \cite{BIK07} concerning adversarial order and random assignment, by providing a constant-factor approximation algorithm in the AO-RA-MK model, and showing that no constant-factor approximation exists
in the AO-RA-MN model.
More precisely, we prove that there is a $40/(1-1/e)$-approximation in the AO-RA-MK model,
i.e. in the model where weights are assigned to the elements of a matroid randomly,
the elements arrive in an adversarial order, and the matroid is known in advance.
We provide a simple thresholding algorithm, which gives a constant-factor approximation  for the AO-RA-MK model when the matroid $\cM$ is uniformly dense. Then we use the principal sequence of a matroid to design a constant-factor approximation for any matroid using the machinery developed by Soto \cite{Soto11}.
(Subsequently to our work, Soto \cite{SotoPhD} improved our approximation factor in the AO-RA-MK model to
$16/(1-1/e)$.)

On the other hand, if the matroid is  not known in advance  (AO-RA-MN model),
we prove that the problem cannot be approximated better than within $\Omega(\log n / \log \log n)$.
This holds even in the special case of rank 1 matroids; see below.
On the positive side, we show an $O(\log r \log n)$-approximation
for this model. 
We achieve this by providing
an $O(\log{r})$-approximation thresholding algorithm for the AO-AA-MU model (when both the input ordering and the assignment of weights to the elements the matroid are adversarial), when an estimate on the weight of the largest non-loop element is given. Here, the novel technique is to employ a dynamic threshold depending on the rank of the elements seen so far.

\paragraph{Classical secretary with unknown $n$.}
A very natural question that arises in this context is the following.
Consider the classical secretary problem, where we want to select $1$ candidate out of $n$.
The classical solution relies on the fact that $n$ is known in advance. However, what if we do not know $n$
in advance, which would be the case in many practical situations?
We show that if an upper bound $N$ on the possible number of candidates $n$ is given
(RA-CN model: i.e., $n$ is chosen by an adversary from $\{1,\ldots,N\}$),
the best candidate can be found with probability $1/(H_{N-1}+1)$,
while there is no algorithm which achieves probability better than $1/H_N$
(where $H_N = \sum_{i=1}^{N} \frac{1}{i}$ is the $N$-th harmonic number).

In the model where we maximize the expected value of the selected candidate, and $n$ is chosen adversarially from $\{1,\ldots,N\}$,
we prove we cannot achieve approximation better than $\Omega(\log N / \log \log N)$.
On the positive side, even if no upper bound on $n$ is given, 
the maximum-weight element can be found with probability $\epsilon / \log^{1+\epsilon} n$ for any fixed $\epsilon>0$.
We remark that similar results follow from \cite{HKS07} and \cite{FHKMS10} where an equivalent problem was considered in the context of online auctions.
More generally, for the matroid secretary problem where no information at all is given in advance (RO-AA-MU),
we achieve an $O(\frac{1}{\epsilon} \log r \log^{1+\epsilon} n)$ approximation for any $\epsilon>0$.
See Table \ref{tab:results} for an overview of our results.

\begin{table}[htb]
\centering
\begin{tabular}{|| c || c | c ||} \hline
\mbox{Problem}
& \mbox{New approximation}
 & \mbox{New hardness}
  \\
\hline \hline
RA-CN &  $H_{N-1}+1$ & $H_N$  \\
\hline
RA-CU &  $O(\frac{1}{\eps} \log^{1+\eps} n)$ & $\Omega(\log n)$  \\
\hline
AO-RA-MK &  $40/(1-1/e)$ & -  \\
\hline
AO-RA-MN &  $O(\log{r} \log{n})$  & $\Omega(\log n / \log\log{n})$  \\
\hline
AO-RA-MU &  $O(\frac{1}{\eps} \log{r} \log^{1+\eps}{n})$  & $\Omega(\log n / \log\log{n})$  \\
\hline
RO-AA-MU &  $O(\frac{1}{\eps} \log{r} \log^{1+\eps}{n})$ &  $\Omega(\log n / \log\log{n})$ \\
\hline
\end{tabular}
\caption{Summary of results}
\label{tab:results}
\end{table}

\paragraph{Organization.}
In section \ref{sec:aoramk} we provide a $40/(1-1/e)$ approximation algorithm for the AO-RA-MK
model. In section \ref{sec:ao-ra-mn} we provide an $O(\log{n}\log{r})$ approximation algorithm for the AO-RA-MN model, and
an $O(\frac{1}{\eps} \log{r} \log^{1+\eps}{n})$ approximation for the RO-AA-MU model. Finally, in section \ref{sec:ra-cn}
we provide a $(H_{N-1}+1)$-approximation and $H_N$-hardness for the RA-CN model.

%
%

\section{Approximation for adversarial order and random assignment}
\label{sec:aoramk}
In this section, we derive a constant-factor approximation algorithm for the AO-RA-MK model,
i.e.  assuming that the ordering of the elements of the matroid is adversarial but weights are assigned
to the elements by a random permutation, and the matroid is known in advance.
We build on Soto's algorithm \cite{Soto11}, in particular on his use of the {\em principal sequence
of a matroid} which effectively reduces the problem to the case of a uniformly-dense matroid while
losing only a constant factor ($1-1/e$). Interestingly, his reduction only requires the randomness in the assignment
of weights to the elements but not a random ordering of the matroid on the input. 
Hence, it is sufficient to obtain a constant factor for uniformly dense matroids. 
Recall that the density of a set in a matroid $\cM=(E,\cI)$ is the quantity
$
\gamma(S) = \frac{|S|}{rank(S)}.
$
A matroid is uniformly dense, if $\gamma(S) \leq \gamma(E)$ for all $S \subseteq E$.
We present a simple thresholding algorithm which works in the AO-RA-MK model (i.e. even for an adversarial ordering
of the elements) for any uniformly dense matroid.
Combining our algorithm with SotoÕs reduction \cite[Lemma 4.4]{Soto11}, we obtain a constant-factor approximation algorithm for the matroid secretary problem in AO-RA-MK model.


Throughout this section we use the following notation.
Let $\cM=(E,I)$ be a uniformly dense matroid of rank $r$. This also means that $\cM$ contains no loops.
Let $|E|=n$ and let $e_1, e_2, \ldots, e_{n}$ denote the ordering of the
elements on the input, which is chosen by an adversary (i.e.~we consider the worst case).
Furthermore, the adversary also chooses $W=\{w_1 > w_2 > \ldots > w_{n}\}$, a set of non-negative weights.
The weights are assigned to the elements of $\cM$ via a random bijection $\weight: E \rightarrow W$.
For a weight assignment $\weight$,
we denote by $w(S) = \sum_{e \in S} \omega(e)$ the weight of a set $S$, and by $\weight(S)  = \{\omega(e): e \in S\}$
the set of weights assigned to $S$.
We also let $\opt(\weight)$ be the maximum-weight independent set in $\cM$.

\subsection{Approximation for uniformly dense matroids}
\label{subsec:uniformdense_aoramk}

We show that there is  a simple thresholding algorithm which includes each of the topmost $\lfloor r/4\rfloor$ weights
(i.e. $w_1,\ldots, w_{\lfloor r/4\rfloor}$) with a constant probability.
This will give us a constant factor approximation algorithm, as $w(\opt(\weight)) \leq \sum_{i=1}^r w_i$, where $w_1>w_2>\ldots>w_r$ are the $r$ largest weights in $W$. It is actually important that we compare our algorithm to the quantity $\sum_{i=1}^{r} w_i$,
because this is needed in the reduction to the uniformly dense case.

The main idea is that the randomization of the weight assignment makes it very likely that the optimum solution contains
many of the top weights in $W$. Therefore, instead of trying to compute the optimal solution with respect to $\omega$,
we can just focus on catching a constant fraction of the top weights in $W$.
Let $A=\{e_1,\ldots,e_{n/2}\}$ denote the first half of the input and $B=\{e_{n/2+1},\ldots,e_n\}$ the second half of the input.
Note that the partition into $A$ and $B$ is determined by the adversary and not random.
Our solution is to use the $\lfloor r/4 \rfloor+1$-st
topmost weight in the "sampling stage" $A$ as a threshold and then include every element in $B$ that is above the threshold
and independent of the previously selected elements.
Details are described in Algorithm \ref{alg:threshold}.

\begin{algorithm}[htb]
\caption{Thresholding algorithm for uniformly dense matroids in AO-RA-MK model}
\label{alg:threshold}
\begin{algorithmic}[1]
\INPUT A uniformly dense matroid $\cM=(E,\cI)$ of rank $r$.
\OUTPUT An independent set $\alg \subseteq E$.
\IF {$r<12$}
\STATE run the optimal algorithm for the classical secretary problem, and return the resulting singleton. \label{alg:largest_element}
\ENDIF
\STATE $\alg \leftarrow \emptyset$
\STATE Observe a half of the input (elements of $A$) and let $w^*$ be the $(\lfloor r/4\rfloor + 1)^{st}$ largest weight among them.\label{alg:threshold:sample}
\FOR {each element $e\in B$ arriving afterwards}
\IF {$\weight(e) > w^*$ and $ALG \cup \{e\}$ is independent}\label{alg:threshold:decide}
\STATE $\alg \leftarrow \alg \cup \{e\}$
\ENDIF
\ENDFOR
\RETURN \alg
\end{algorithmic}
\end{algorithm}

\begin{theorem}
\label{thm:AO-RA-KM}
Let $\cM$ be a uniformly dense matroid of rank $r$, and $\alg(\weight)$ be the set returned by Algorithm \ref{alg:threshold}
when the weights are defined by a uniformly random bijection $\weight: E \rightarrow W$. Then
$$
\EE{\weight}{w(\alg(\weight))} \geq 
\frac{1}{40} \sum_{i=1}^{r} w_i
$$
where $\{w_1>w_2>\ldots>w_r\}$ are the $r$ largest weights in $W$.
\end{theorem}

If $r<12$, the algorithm finds and returns the largest weight $w_1$ with probability $1/e$ (step \ref{alg:largest_element};
the optimal algorithm for the classical secretary problem).
Therefore, for $r<12$, we have $\EE{\weight}{w(\alg(\weight))} \geq \frac{1}{11e} \sum_{i=1}^{r} w_i
 > \frac{1}{40} \sum_{i=1}^{r} w_i$.

For $r \geq 12$, we prove that each of the topmost $\lfloor r/4 \rfloor$ weights will be included in $\alg(\weight)$
with probability at least $1/8$.  Hence, we will obtain
\begin{equation}
\label{eq:larger_algopt_comparison}
\EE{\weight}{w(\alg(\weight))} \geq \frac{1}{8}  \sum_{i=1}^{\lfloor r/4\rfloor} w_i \geq
\frac{1}{40} \sum_{i=1}^{r} w_i.
\end{equation}

Let $t = 2 \lfloor r/4 \rfloor + 2$. Define $C'(\weight) = \{e_j: \weight(e_j) \geq w_t \}$ to be the set of elements of $\cM$
which get one of the top $t$ weights. Also let $A'(\weight) = C'(\weight)\cap A$ and $B'(\weight) = C'(\weight) \cap B$.
Moreover, for each $1 \leq i \leq t$ we define $C'_i(\weight) = \{e_j: \weight(e_j) \geq w_t \ \& \ \weight(e_j) \neq w_i \}$,
$A'_i(\weight) = C'_i(\weight)\cap A$ and $B'_i(\weight) = C'_i(\weight) \cap B$, i.e.~the same sets with the element of weight
$w_i$ removed.

First, we fix $i \leq \lfloor r/4 \rfloor$ and argue that the size of $B'_i(\weight)$ is smaller than $A'_i(\weight)$ with probability $1/2$.
Then we will use the uniformly dense property of $\cM$ to show that the span of $B'_i(\weight)$ is also quite small
with probability $1/2$ and consequently $w_i$ has a good chance of being included in $\alg(\weight)$.

\begin{clm}
\label{cl:Bstar_size}
Let $\cM$ be a uniformly dense matroid of rank $r$, $t = 2 \lfloor r/4 \rfloor + 2$, $1 \leq i \leq \lfloor r/4 \rfloor$,
and $B'_i(\weight)$ defined as above. Then we have
\begin{equation}
\label{eq:Bstar_size}
\PP{\weight}{|B'_i(\weight) | \leq \lfloor r/4 \rfloor} = 1/2.
\end{equation}
\end{clm}

\begin{proof}
Consider $C'_i(\weight)$, the set of elements receiving the top $t$ weights except for $w_i$.
This is a uniformly random set of odd size $t-1 = 2 \lfloor r/4 \rfloor + 1$.
By symmetry, with probability exactly $1/2$, a majority of these elements are in $A$,
and hence at most $\lfloor r/4 \rfloor$ of these elements are in $B$,
i.e. $|B'_i(\weight)| \leq \lfloor r/4 \rfloor$.
\end{proof}

Now we consider the element receiving weight $w_i$.
We claim that this element will be included in $\alg(\weight)$ with a constant probability.

\begin{clm}
\label{cl:constant_prob}
Let $\cM$ be a uniformly dense matroid of rank $r$, and $i \leq \lfloor r/4 \rfloor$. Then
$$
\PP{\weight}{\weight^{-1}(w_i) \in \alg(\weight)} \geq 1/8.
$$
\end{clm}

\begin{proof}
Condition on $C'_i(\weight) = S$ for some particular set $S$ of size $t-1$ such that $|B'_i(\weight)| =
 |S \cap B| \leq \lfloor r/4 \rfloor$.
This fixes the assignment of the top $t$ weights except for $w_i$.
Under this conditioning, weight $w_i$ is still assigned uniformly to one of the remaining $n-t+1$ elements.

Since we have $|A'_i(\weight)| = |S \cap A| \geq \lfloor r/4 \rfloor + 1$, the threshold $w^*$ in this case
is one of the top $t$ weights and the algorithm will never include any weight outside of the top $t$.
Therefore, we have $\alg(\weight) \subseteq B'(\weight)$.
The weight $w_i$ is certainly above $w^*$ because it is one of the top $\lfloor r/4 \rfloor$ weights.
It will be added to $\alg(\weight)$ whenever it appears in $B$ and
it is not in the span of previously selected elements. Since all the previously included elements must be in $B'_i(\weight) = S \cap B$,
it is sufficient to avoid being in the span of $S \cap B$. To summarize, we have
\begin{eqnarray}
\weight^{-1}(w_i) \in B \setminus span(S \cap B)
 & \Rightarrow & \weight^{-1}(w_i) \in \alg(\weight). \nonumber
\end{eqnarray}
What is the probability that this happens?
Similar to the proof of \cite[Lemma 3.1]{Soto11}, since $\cM$ is uniformly dense, we have
$$
\frac{|span(S \cap B)|}{|S \cap B|} \leq \frac{|span(S \cap B)|}{rank(span(S \cap B))} \leq \frac{n}{r}
 \Longrightarrow |span(S \cap B)| \leq \frac{n}{r} |S \cap B| \leq \frac{n}{4}
$$
using $|S \cap B| \leq \lfloor r/4 \rfloor$. Therefore, there are at least $n/4$ elements in $B \setminus span(S \cap B)$.
Given that the weight $w_i$ is assigned uniformly at random among $n-t$ possible elements,
we get
$$ \PP{\weight}{\weight^{-1}(w_i) \in B \setminus span(S \cap B) \mid C'_i(\weight) = S} \geq \frac{n/4}{n-t} \geq \frac{1}{4}.$$
Since this holds for any $S$ such that $|S \cap B| \leq \lfloor r/4 \rfloor$, and $S \cap B = C'_i \cap B = B'_i(\weight)$,
it also holds that
$$ \PP{\weight}{\weight^{-1}(w_i) \in B \setminus span(B'_i(\weight)) \mid |B'_i(\weight)| \leq \lfloor r/4 \rfloor} \geq \frac{1}{4}. $$
Using Claim~\ref{cl:Bstar_size}, we get $\PP{\weight}{\weight^{-1}(w_i) \in B \setminus span(B'_i(\weight))} \geq 1/8$.
\end{proof}

This finishes the proof of Theorem~\ref{thm:AO-RA-KM}.


\subsection{Extension  to general matroids}
\label{app:principal_minor}
In this section we describe the final $40/(1-1/e)$ approximation algorithm for AO-RA-MK model for general matroids. The algorithm is based on Soto's algorithm \cite{Soto11}, by decomposing the underlying matroid into a sequence of principal minors and then running Algorithm \ref{alg:threshold} in parallel on each of them separately.

\begin{algorithm}[htb]
\caption{Thresholding algorithm for matroid secretary problem in AO-RA-MK model}
\label{alg:general_aoramk}
\begin{algorithmic}[1]
\INPUT A matroid $\cM=(E,\cI)$.
\OUTPUT An independent set $\alg \subseteq E$.
\STATE Compute the sequence of principal minors $(\cM_i)_{i=1}^k$:
Initialize $k=0$. While $\bigcup_{i=1}^{k} E_i \neq E$, let $E_{k+1}$ be the densest set in the matroid $\cM / \bigcup_{i=1}^{k} E_i$,
define $\cM_{k+1} = (\cM / \bigcup_{i=1}^{k} E_i) | E_{k+1}$, and increment $k$.
\STATE Run Algorithm \ref{alg:threshold} in parallel on each $\cM_i$ to get a solution $I_i$, and return $ALG = \bigcup_{i=1}^{k} I_i$.
\end{algorithmic}
\end{algorithm}

We use Soto's lemma to argue that if the weights are assigned randomly to the elements, and we achieve
an $\alpha$-fraction of the sum of the $r_i$ topmost weights in each principal minor $\cM_i$,
then we obtain an $\alpha/(1-1/e)$ approximation overall.

Interestingly, it is necessary to know the matroid in advance, in order to discriminate the  dense parts of the matroid from the sparse parts (by computing the principal minors), and try to handle them separately. Otherwise,  as we prove later,
no algorithm can do better than an $O(\log n/ \log\log n)$ approximation.


\begin{corollary}
Algorithm \ref{alg:general_aoramk} gives a $\frac{40}{1-1/e}$-approximation in the AO-RA-MK model.
\end{corollary}

\begin{proof}
Similar to the proof of Theorem \ref{thm:AO-RA-KM}, let $e_1,\ldots,e_n$ be the sequence of elements of $M$ designed by the adversary and $W=w_1 > \ldots > w_n$ be the hidden list of weights.
Let $\cM_i$, $1\leq i\leq k$, be the sequence of principal minors of $\cM$, with ground set $E_i$ and rank $r_i$,
and let $\cP$ denote a partition matroid as defined in \cite[Section 4]{Soto11}, with ground set $E$ and independent sets
$$
\cI(\cP) = \left\{ \bigcup_{i=1}^k I_i: I_i \subseteq E_i,|I_i|\leq r_i \right\}.
$$
For a uniformly random bijection $\weight:E\rightarrow W$, let $\opt_{\cP}(\weight)$ be the maximum weight of an independent
set in matroid $\cP$, and $\alg(\weight)$ be the set returned by Algorithm \ref{alg:general_aoramk}. Conditioning on the set of weights assigned to the elements of each block $E_i$, the elements in $E_i$ receive a random permutation of this set of weights.
Since each  $\cM_i$ is uniformly dense, By  Theorem \ref{thm:AO-RA-KM}, Algorithm \ref{alg:general_aoramk} recovers in expectation
a $1/40$-fraction of the sum of the heaviest $r_i$ weights assigned to elements in $E_i$.
However, the union of the heaviest $r_i$ elements in each $E_i$ is indeed the optimum solution in the partition matroid $\cP$.
By removing the conditioning we get
\begin{equation}
\label{eq:comparing_partition}
\EE{\weight}{w(\alg(\weight))} \geq \frac{1}{40} \EE{\weight}{\opt_{\cP}(\weight)}.
\end{equation}
Moreover, Soto in \cite{Soto11} proved that  $\EE{\weight}{\opt_{\cP}}$ is only a constant factor away from the optimum of $\EE{\weight}{\opt_{\cM}}$.
\begin{lemma}[Soto \cite{Soto11}]
\label{lem:sotoprincipal_minor}
$\EE{\weight}{w(\opt_{\cP})(\weight))} \geq (1-1/e) \EE{\weight}{w(\opt_{\cM}(\weight))}.$
\end{lemma}
This proves the corollary.
\end{proof}


\section{Approximation algorithms for unknown matroids}
\label{sec:ao-ra-mn}
In this section we will be focusing mainly on the AO-RA-MN model.
 i.e. assuming that the ordering of the elements of the matroid is adversarial, weights are assigned randomly, but the matroid is unknown, and the algorithm only knows $n$ in advance. We present an $O(\log{n}\log{r})$ approximation algorithm for the AO-RA-MN model,
where $n$ is the number of elements in the ground set and $r$ is the rank of the matroid.
It is worth noting that in these models the adversary may set some of the elements of the matroid to be loops, and the algorithm does
not know the number of loops in advance. For example it might be the case that after observing the first 10 elements, the rest are all loops and thus the algorithm should select at least one of the first 10 elements with some non-zero probability. This is the idea of the counterexample
in section \ref{sec:ra-cn} (Corollary \ref{cor:hardness_aoramn}), where we reduce AO-RA-MN, AO-RA-MU models to RA-CN, RA-CU models respectively, and thus we show that there is no constant-factor approximation for either of the  models. In fact, no algorithm can do better than $\Omega(\log{n}/\log\log{n})$. Therefore, our algorithms are tight within a factor of $O(\log{r}\log\log{n})$ or $O(\log{r}\log^{\eps}{n})$.

  At the end of this section we  also give a general framework that can turn any $\alpha$ approximation algorithm for the RO-AA-MN model, (i.e. the primary variant of the matroid secretary problem) into an $O(\alpha \log^{1+\eps}{n}/\eps)$ approximation algorithm in the RO-AA-MU model (see subsection \ref{sec:RO-AA-MU}).

We use the same notation as section \ref{sec:aoramk}: $\cM=(E,I)$ is a  matroid of rank $r$ (which is not known to the algorithm), and $e_1,e_2,\ldots,e_n$ is the the adversarial ordering of the elements of $\cM$, and $W=\{w_1 > w_2 > \ldots >w_n\}$ is the set of hidden weights chosen by the adversary that are assigned to the elements of $\cM$ via a random bijection $\weight: E \rightarrow W$.

\subsection{Approximation for AO-RA-MN models}
\label{subsec:aoramn_aoramu}
We start by deriving an $O(\log{n}\log{r})$ approximation algorithm for the AO-RA-MN model.
Our algorithm basically tries to ignore the the loops and only focuses on the non-loop elements.
We design our algorithm in two phases.
In the first phase we design a randomized algorithm that works even in the AO-AA-MU model assuming that it has a good estimate on the weight of the largest non-loop element. In particular, fix bijection $\weight: W\rightarrow E$, and let $e^*_1$ be the largest non-loop element with respect to $\weight$, and $e^*_2$ be the second largest one. We assume that the algorithm knows a bound  $\weight(e^*_2) < \L < \weight(e^*_1)$ on the largest non-loop element in advance. We show there is a thresholding algorithm, with a {\em non-fixed} threshold, that achieves an $O(\log{r})$ fraction of the optimum (see subsection \ref{subsubsec:aoramn_bound}).

%

In order to solve the original problem, in the second phase we divide the non-loop elements into a set of blocks $B_{1}, B_{2}, \ldots, B_{\log{n}}$, and we use the previous algorithm as a module to get an $O(\log{r})$ of optimum within each block (see subsection \ref{subsubsec:aoramn_general}).

\subsubsection{Approximation for AO-RA-MN model, with an estimate on the largest weight}
\label{subsubsec:aoramn_bound}
Let us start by the first phase.  Since our algorithm works in a more general model, here we assume that we are in the AO-AA-MU model, i.e. assuming that both the ordering of the elements and assignments of the weights are chosen adversarially, and the algorithm knows nothing except a bound $\weight(e^*_2) < \L< \weight(e^*_1)$ on the largest non-loop element.
  We design a randomized $O(\log{r})$ approximation algorithm for this model.

Note that if $r$ is also known in advance then a simple variant of the  thresholding algorithm of Babaioff  et al. \cite[ThresholdPrice Algorithm]{BIK07} would be a $O(\log{r})$ approximation. Indeed it is sufficient  to select a threshold $\L/2^{i}$, for $0\leq i\leq \log{r}$ uniformly at random, and then include all the elements above the threshold that are independent of the elements chosen so far. Here, since we do not know $r$, our algorithm keeps track of the rank of the elements seen so far, and tries to update the threshold according to it. In particular, once the rank of the elements seen so far reaches $2^{i}$, the algorithm inserts a new threshold dynamically and works with it as if it exists since the beginning of the algorithm. The details are described in Algorithm \ref{alg:threshold_aoramn}:

\begin{algorithm}[htb]
\caption{Algorithm for AO-AA-MU model, when an estimate of the largest non-loop element is known}
\label{alg:threshold_aoramn}
\begin{algorithmic}[1]
\INPUT The bound $L$ such that $\weight(e^*_2)< \L< \weight(e^*_1)$.
\OUTPUT An independent set $\alg \subseteq E$.
\STATE with probability $1/2$, pick a non-loop element with weight above $\L$ and return it. \label{alg:step:largest_element}
\STATE $\alg \leftarrow \emptyset$ and $r^* \leftarrow 2$.
\STATE set threshold $w^* \leftarrow \L/2$.
\FOR {each arriving element $e_i$}
\IF {$\weight(e_i) > w^*$ and $ALG \cup \{e_i\}$ is independent}
\STATE $\alg \leftarrow \alg \cup \{e_i\}$
\ENDIF
\IF {$rank(\{e_1,\ldots,e_i\}) \geq r^*$}
\STATE with probability $\frac{1}{\log 2r^*}$ set $w^* \leftarrow \L/2r^*$.\label{alg:step:probability_update}
\STATE $r^* \leftarrow 2r^*$.
\ENDIF
\ENDFOR
\RETURN \alg
\end{algorithmic}
\end{algorithm}
\def\Ef{{\cal E}_1}
Let $\Ef$ be the event the algorithm chooses the option in step \ref{alg:step:largest_element}.
Also let $r^*(t)$ and $w^*(t)$ be the value of $r^*$ and $w^*$, respectively, after observing the first $t$ elements of the input. In particular, $r^*(n)$ will be the rank of $\cM$, and $w^*(n)$ will be the final value of the threshold chosen by the algorithm.
The following observation describes some properties of the algorithm:

\begin{observation}
\label{obs:alg:aoramn}
Assuming $\neg\Ef$, for any matroid of rank $r$, observe that  $r^*(n)$ in the algorithm will be the smallest power of 2 greater than $r$ (i.e. $r^*(n) \leq 2r$). Therefore, the algorithm will choose between at most $\log{(2r)}$ different thresholds, where for each $i$, the threshold $w^*(t)$ will be decreased to $\L/2^i$  at the first time $t(i)$ where $rank(e_1,\ldots,e_{t(i)}) = 2^{i-1}$, with probability $1/i$.

Hence, by applying a simple induction it is not hard to see that at any time $t$ in the execution of the algorithm,
\begin{equation}
\label{eq:equalprob_threshold}
1 \leq i \leq \log{r^*(t)}, ~~\P{w^*(t)=\frac{\L}{2^i} \Big| \neg\Ef}=1/\log{r^*(t)},
\end{equation}
where the probability is over all of the randomization in the algorithm.
\end{observation}

\begin{theorem}
\label{thm:alg:aoramn_L}
For any  matroid $\cM=(E,\cI)$ of rank $r$, and any bijection $\weight:E\rightarrow W$, given the bound  $\weight(e^*_2) <\L< \weight(e^*_1)$, Algorithm \ref{alg:threshold_aoramn} is a $16\log{r}$ approximation in the AO-AA-MU model. i.e.
$$\E{w(\alg(\weight))} \geq \frac{1}{16\log{r}} w(\opt(\weight)),$$
where the expectation is over all of the randomization in the algorithm.
\end{theorem}

Let us partition the elements of $\opt(\weight)$ according to their weights, where
 \begin{equation}
 \label{eq:pi_definition}
1\leq i\leq\log{2r}:~~ P_i=\left\{e\in \opt(\weight): \frac{\L}{2^{i}} < \weight(e) \leq \frac{\L}{2^{i-1}}\right\}.
 \end{equation}

First in the next claim, we show that conditioned on $w^*(n)=L/2^i$ (and $\neg\Ef$), the expected weight of $\alg(\weight)$, is a constant fraction of $w(P_i)$, unless the size of $|P_i|$ is very small. In the latter case as we will show in equation \eqref{eq:chargebyw1}, we may charge $w(P_i)$ by a $1/\log{r}$ fraction of $\weight(e^*_1)$. Since $\Ef$ occurs with constant probability, the algorithm achieves a constant fraction of $\weight(e^*_1)$ which completes the proof.

\begin{clm}
\label{cl:large_pi}
For any $1\leq i\leq\log{r}$, if $|P_{i}| \geq 2^i$, then
$$\E{w(\alg(\weight)) | w^*(n)=\frac{L}{2^i} \wedge \neg\Ef } \geq \frac{1}{4} w(P_{i}).$$
\end{clm}
\begin{proof}
Let $E_i=\{e_1, \ldots, e_{i}\}$ be the set of the first $i$ elements. Recall that $t(i)$
is the first time $t$ where $rank(E_{t}) = 2^{i-1}$.
Since $P_i\subseteq \opt(\weight)$ is an independent set of $\cM$,  
we have $|P_i \cap E_{t(i)}| \leq 2^{i-1}$. In other words, we must have seen at most $2^{i-1}$ elements of the set $P_i$ by the time $t(i)$.

Suppose $w^*(n)=L/2^i$; since  $w^*(t)$ is a non-increasing function of $t$ (with probability 1), we get $w^*(t)\geq L/2^i$.
Since  $P_i\setminus E_{t(i)}$ is an independent set and all its elements will come after $t(i)$, we get $|\alg(\weight)| \geq |P_i\setminus E_{t(i)}| \geq |P_i| - 2^{i-1}$ by the end of the algorithm. But all these elements are greater than $w^*(n)=L/2^i$, thus:
$$
\E{w(\alg(\weight)) | w^*(n)=\frac{L}{2^i} \wedge \neg\Ef}
\geq  |P_i\setminus E_{t(i)}| \frac{L}{2^i}
\geq \frac{|P_i| L }{2^{i+1}} \geq
\frac{1}{4} w(P_i),
$$
where the last inequality follows from equation \eqref{eq:pi_definition}.
\end{proof}



Now we are ready to prove Theorem \ref{thm:alg:aoramn_L}

\begin{proofof}{Theorem \ref{thm:alg:aoramn_L}}
Using the above claim we may simply compute the overall performance of the algorithm:
\begin{eqnarray}
\E{w(\alg(\weight))} &=& \frac{1}{2}\E{w(\alg(\weight)) | \Ef} + \frac{1}{2}\E{w(\alg(\weight)) | \neg\Ef} \nonumber \\
&\geq & \frac{1}{2} \weight(e^*_1) + \frac{1}{2} \sum_{i:|P_i| \geq 2^i}  \E{w(\alg(\weight)) \Big| w^*(n) = \frac{\L}{2^i}\wedge\neg\Ef} \P{w^*(n) = \frac{L}{2^i} \Big| \neg \Ef} \nonumber \\
&\geq& \frac{\weight(e^*_1)}{2} + \frac{1}{2} \sum_{i:|P_i|\geq 2^i} \frac{w(P_i)}{4} \frac{1}{\log{2r}} \label{eq:boundlargepi}\\
&\geq& \frac{\weight(e^*_1)}{4} + \sum_{i=1}^{\log{2r}} \frac{2L}{8\log{2r}} +\frac{1}{2} \sum_{i:|P_i|\geq 2^i} \frac{w(P_i)}{4} \frac{1}{\log{2r}} \label{eq:chargebyw1}\\
&\geq & \frac{\weight(e^*_1)}{4} + \sum_{i=1}^{\log{2r}} \frac{w(P_i)}{8\log{2r}},\label{eq:boundeverythingbyL}
\end{eqnarray}
where   inequality \eqref{eq:boundlargepi} follows from equation \eqref{eq:equalprob_threshold} and Claim \ref{cl:large_pi}, inequality \eqref{eq:chargebyw1} follows from the assumption $\weight(e^*_1)\geq L$, and   inequality \eqref{eq:boundeverythingbyL} follows from $w(P_i)\leq |P_i|\frac{L}{2^{i-1}} \leq 2L$ for $|P_i|\leq 2^i$.

The theorem  simply follows from the fact that $w(\opt(\weight))\leq 2(\weight(e^*_1) + \sum w(P_i))$.
\end{proofof}

Before describing our algorithm for the AO-RA-MN model, we  prove a bound on the performance
of  algorithm \ref{alg:threshold_aoramn} when the bound $L$ can be much larger than the maximum weight (i.e. $\weight(e^*_1) \ll L)$. This  may happen as a special case when we want to apply Algorithm \ref{alg:threshold_aoramn} as a subroutine.
\begin{corollary}
\label{cor:moreon_threshold_aoramn}
For any matroid $\cM=(E,\cI)$ of rank $r$, and any bijection $\weight:E \rightarrow W$, given any bound $L > \weight(e^*_2)$ we have
\begin{equation}
\label{eq:largeL_case}
\E{w(\alg(\weight))} \geq \max\left(0, \frac{w(\opt(\weight))}{16\log{r}} - 2\L\right).
\end{equation}
If in addition $L < \weight(e^*_1)$, then
\begin{equation}
\label{eq:Lpicking_prob}
\E{w(\alg(\weight))} \geq \frac{\weight(e^*_1)}{2}.
\end{equation}
\end{corollary}

\begin{proof}
To prove the first inequality, note that if $\L<\weight(e^*_1)$, then we are done, otherwise suppose that we increase the weight of $e^*_1$ to $L+\weight(e^*_1)$.  Define $\weight'=\weight$ on all elements, except $\weight'(e^*_1)=L+\weight(e^*_1) \leq 2L$. Then by Theorem \ref{thm:alg:aoramn_L}, we have
$$
\E{w(\alg(\weight'))} \geq \frac{w(\opt(\weight'))}{16\log{r}} = \frac{w(\opt(\weight))+L}{16\log{r}}
$$
On the other hand, since in the worst case $\alg(\weight)$  does not have $e^*_1$, while $\alg(\weight')$ has it, we have $\E{w(\alg(\weight))} \geq \E{w(\alg(\weight'))} - 2L$.  Therefore
$$
\E{w(\alg(\weight))} \geq \frac{w(\opt(\weight))+\L}{16\log{r}} - 2\L \geq \max\left(0, \frac{w(\opt(\weight))}{16\log{r}} - 2\L\right).
$$
 The second inequality can be proved simply by noting that  the algorithm picks $e^*_1$ in step \ref{alg:step:largest_element} with  probability 1/2.
\end{proof}

\subsubsection{Approximation for AO-RA-MN by a general reduction}
\label{subsubsec:aoramn_general}

Now we are ready to describe our final algorithm for AO-RA-MN model without knowing $\L$ in advance (here, unlike the previous algorithm we will use the random assignment of weights). The idea is to only consider the non-loop elements and divide them into a set of blocks $B_1,B_2,\ldots,B_{\log{2n}}$ such that $|B_i|=2^i$ (note that the number of non-loop elements can be quite smaller than $n$, but we do not know it in advance). After observing the first $i$ blocks, we would have a good guess on the largest weight of the next block. Using that guess as a bound $\L$, with probability $1/\log{(2n)}$, we  run Algorithm \ref{alg:threshold_aoramn} on block ${i+1}$  and return its solution as the final answer. The details are described in Algorithm \ref{alg:aoramn}.

\begin{algorithm}[htb]
\caption{Algorithm for AO-RA-MN model}
\label{alg:aoramn}
\begin{algorithmic}[1]
\INPUT $n$, the number of elements.
\OUTPUT An independent set $\alg \subseteq E$.
\STATE Choose a number $0\leq b\leq \log{n}$ uniformly at random. \label{alg:step:randomblock}
\STATE Observe the first $2^b -1$ non-loop elements without picking any of them, and let $\L(b)$ be the largest weight among these non-loop elements.
\STATE Run Algorithm \ref{alg:threshold_aoramn} only on the next $2^b$ non-loop elements (ignore loops), with parameters $n=2^b$ and $\L=\L(b)$, and return its output.
\end{algorithmic}
\end{algorithm}

The next theorem proves the correctness of the algorithm
\begin{theorem}
\label{thm:aoramn}
For any matroid $\cM=(E,\cI)$ of rank $r$, Algorithm \ref{alg:aoramn} is a $O(\log{r}\log{n})$ approximation in the AO-RA-MN model.
\end{theorem}
Let  $F$ be the set of non-loop elements, $m:=|F|$, and let $F_i\subset F$ be the set of first $2^{i+1}-1$ non-loop elements (as a special case $F_{\log{m}}=F$. We divide the elements of $F$ into a set of blocks $B_0, B_1, \ldots, B_{\lfloor \log{m}\rfloor }$,  where $B_0:=F_0$, and for each $i>0$, $B_i:=F_i\setminus F_{i-1}$. Note that the size of the last block $|B_{\log{m}}| = m+1-2^{\lfloor \log{m} \rfloor}$ can be much smaller than $2^{\lfloor \log{m}\rfloor}$.

For a set of weights $W'\subset W$ and $E'\subset E$ of elements such that $|W'|=|E'|$, let $\cE_{W'}(E')$ be the event $\weight(E')=W')$.
Fix a set $W'\subset W$ of size $|W'|=|F|$. 
Throughout the proof we always condition on $\cE_{W'}(F)$. 
Define
\begin{equation}
0\leq i\leq \lfloor \log{m}\rfloor : ~~~O_i=\EE{\weight}{w(\opt(\weight) \cap B_i) | \cE_{W'}(F)},
\end{equation}
to be the expected value of the optimum set in each of the blocks.
We will show that
$$\EE{\weight}{w(\alg(\weight)) | \cE_{W'}(F)} \geq \frac{1}{2500 \log{r}\log{n}} \sum_{i=1}^{\log{m}} B_i.$$
In the next claim we show that conditioned on algorithm chooses $b=i$ in the step \ref{alg:step:randomblock}, it will get an $\Omega(1/\log{r})$ fraction of $O_{i}$.
Note that in this claim we do not analyze the special case of $b=\lfloor \log{m}\rfloor$.
\begin{clm}
\label{cl:Oi_bound}
If the algorithm chooses $b=i < \lfloor \log{m}\rfloor $ in step \ref{alg:step:randomblock}, it will get an $\Omega(1/\log{r})$ fraction of $O_i$:
$$
\EE{\weight}{w(\alg(\weight)) | b=i, {\cal E}_{W'}(F)} \geq \frac{1}{128\log{r}}O_i.
$$
\end{clm}
\begin{proof}
Fix a set of weights $S=\{s_1>s_2>\ldots>s_{2^{i+1}-1}\} \subset W'$. Conditioned on $\cE_{S}(F_i)$, there is a constant probability that $s_1\in \weight(B_{i})$ and $s_2\notin \weight(B_{i})$; thus  $\L(b)=s_2$ will be a feasible bound for Algorithm \ref{alg:threshold_aoramn}.
Therefore, we may apply Theorem \ref{thm:alg:aoramn_L} and obtain $\Omega(\log{r})$ fraction of $O_i$. Thus
\begin{eqnarray}
&&\EE{\weight}{w(\alg(\weight)) |  \cE_{S}(F_i), b=i, \cE_{W'}(F)} \geq \nonumber \\
&&~~~~\geq \frac{1}{4}\EE{\weight}{w(\alg(\weight)) | s_1 \in B_{i}, s_2\notin B_{i}, \cE_{S}(F_i), b=i, \cE_{W'}(F)} \nonumber \\
&&~~~~\geq \frac{1}{64\log{r}}\EE{\weight}{w(\opt(\weight)\cap B_{i}) | s_1 \in B_{i}, s_2\notin B_{i}, \cE_{S}(F_i),  \cE_{W'}(F) }   \label{eq:usingthresholdaoramn}  \\
&&~~~~\geq \frac{1}{128\log{r}}\EE{\weight}{w(\opt(\weight)\cap B_{i}) |\cE_{S}(F_i), \cE_{W'}(F)}. \label{eq:optlowerbound}
\end{eqnarray}
Here inequality \eqref{eq:usingthresholdaoramn} follows from Theorem \ref{thm:alg:aoramn_L}, and
inequality \eqref{eq:optlowerbound}
    holds by noting that removing the condition $s_2 \notin B_{i}$ can only double the expectation of OPT, while removing $s_1\notin B_{i}$ may only decrease its expectation.
The claim simply follows by summing up inequality \eqref{eq:optlowerbound} over all events $\cE_S(F_i)$, for any $S\subset W', |S|=2^{i+1}-1$.

\end{proof}

Now we are ready to Prove Theorem \ref{thm:aoramn}

\begin{proofof}{Theorem \ref{thm:aoramn}}
We use Claim \ref{cl:Oi_bound} to lower bound the expected gain of the algorithm from all except the last block. We need to analyze $b=\lfloor \log{m}\rfloor$ differently.
Indeed if $B_{\lfloor \log{m}\rfloor} \ll m/2$, the bound $L(b)$ will be much larger than the largest weight in $\weight(B_{\lfloor \log{m}\rfloor})$ w.h.p. Therefore, we  apply
Corollary \ref{cor:moreon_threshold_aoramn} for this special case. Intuitively, the loss
incurs by misreporting the bound $L(\lfloor \log{m}\rfloor)$ is no more than the largest weight in $W'$, and this can be compensated simply by selecting the largest weight with constant probability.


Let $L'$ be the largest weight in $W'$.
By Corollary \ref{cor:moreon_threshold_aoramn} (equation \eqref{eq:largeL_case}), we obtain
$$
\EE{\weight}{w(\alg(\weight)) | b=\lfloor \log{m} \rfloor, {\cal E}_{W'}(F)} \geq \max\left(0, \frac{O_{\lfloor \log{m} \rfloor}}{16\log{r}} - 2\L'\right).
$$
Therefore, by Claim \ref{cl:Oi_bound} and the above inequality  we get:
\begin{eqnarray}
\EE{\weight}{w(\alg(\weight))} & =& \sum_{i=0}^{\lfloor \log{m} \rfloor} \EE{\weight}{w(\alg(\weight)) | b=i, {\cal E}_{W'}(F)} \PP{\weight}{b=i|{\cal E}_{W'}(F)} \nonumber \\
\label{eq:algOi_bounds}
&\geq& \frac{1}{\log{2n}} \left(\sum_{i=0}^{\lfloor \log{m} \rfloor-1} \frac{O_i}{128\log{r}} +  \max\{0, \frac{O_{\lfloor \log{m} \rfloor}}{16\log{r}} - 2\L'\} \right).
\end{eqnarray}

In order to lower bound the RHS it suffices to show that $\EE{\weight}{w(\alg(\weight)) | {\cal E}_{W'}} = \Omega(L'/\log{n})$. This simply follows from the second part of Corollary \ref{cor:moreon_threshold_aoramn}. For any block $B_i$, conditioned on $L'\in \weight(B_i)$,
 with probability $1/2$, the second largest weight in $\weight(F_i)$, is not assigned to $B_i$, in which case algorithm achieves $L'$ with probability $1/2$, once it chooses $b=i$:
\begin{eqnarray}
 \EE{\weight}{w(\alg(\weight)) | \cE_{W'}(F)} &=& \sum_{i=0}^{\log{m}}\frac{|B_i|}{\log{m}}\EE{\weight}{w(\alg(\weight)) | L'\in \weight(B_i), \cE_{W'}(F)}   \nonumber \\
  &=& \sum_{i=0}^{\log{m}} \frac{|B_i|}{\log{m}} \frac{\EE{\weight}{w(\alg(\weight)) | b=i, L'\in \weight(B_i), \cE_{W'}(F)}}{\log{2n}}\nonumber \\
 &\geq & \sum_{i=0}^{\log{m}}  \frac{|B_i|L'}{4\log{m}\log{2n}} = \frac{L'}{4\log{2n}}.
 \label{eq:algL_bounds}
 \end{eqnarray}


Therefore, by adding up equation \eqref{eq:algOi_bounds} and 8 times equation $\eqref{eq:algL_bounds}$ we obtain
\begin{eqnarray*}
 9\EE{\weight}{w(\alg(\weight)) |  {\cal E}_{W'}}& \geq &\frac{1}{\log{2n}} \left(\sum_{i=0}^{\lfloor \log{m} \rfloor-1} \frac{O_i}{128\log{r}} +  \max\{0, \frac{O_{\lfloor \log{m} \rfloor}}{16\log{r}} - 2\L'\} +2L'\right) \\
&\geq& \sum_{i=0}^{\lfloor \log{m} \rfloor} \frac{O_i}{128\log{r}\log{2n}} =
\Omega\left(\frac{1}{\log{r}\log{n}}\right) \EE{\weight}{w(\opt(\weight)) | {\cal E}_{W'}}.
\end{eqnarray*}
Summing both sides of the inequality over all events ${\cal E}_{W'}$ completes the proof.
\end{proofof}


\subsection{Matroid secretary with unknown $n$}
\label{sec:RO-AA-MU}

In this subsection we consider the primary variant of the matroid secretary problem.
When the total number of elements $n$ is known in advance
(RO-AA-MN model), there is  an $O(\log r)$-approximation which was designed in \cite{BIK07}
and is still the best known approximation for this problem.

Here we show a simple reduction which implies that if we do not have any information about the matroid or the number of elements
(the RO-AA-MU model), we can achieve an $O(\frac{1}{\eps} \log^{1+\eps} n \log r)$-approximation for any fixed $\eps>0$.

\begin{theorem}
Let $\cM$ be a matroid of rank $r$ on $n$ elements. If there is  an
$\alpha$ approximation algorithm for the matroid secretary
problem on $\cM$ in the RO-AA-MN model, then for any fixed $\eps>0$, there is also an $O(\frac{\alpha}{\eps} \log^{1+\eps} n )$-approximation
for the matroid secretary problem on $\cM$ with no information given in advance (the RO-AA-MU model).
\end{theorem}

\begin{proof}
We guess a number $n'$ according to a probability distribution with a polynomial tail, as follows:
let $n' = 2^i$ where $i \geq 0$ is chosen with probability
$$ p_i = \frac{\eps}{1+\eps} \cdot \frac{1}{(1+i)^{1+\eps}}.$$
This distribution is chosen so that $\sum_{i=1}^{\infty} p_i \leq 1$ (with the remaining probability, we do nothing);
this can be verified as follows:
$$ \sum_{i=0}^{\infty} \frac{1}{(1+i)^{1+\eps}} = 1 + \sum_{i=1}^{\infty} \frac{1}{(1+i)^{1+\eps}}
 \leq 1 + \int_0^\infty \frac{dx}{(1+x)^{1+\eps}} = 1 + \left[ - \frac{1/\eps}{(1+x)^\eps} \right]_0^\infty = 1 + \frac{1}{\eps}.$$
Then we run the $\alpha$-approximation algorithm as a black box, 
under the assumption that the number of elements is $n'$.

Assume that the actual number of elements is $n \in [2^i, 2^{i+1})$. With probability $p_i$, our guess of the number
of elements is $n' = 2^i$. If this happens, we retrieve $1/\alpha$ of the expected value of the optimal solution
on the first $n'$ elements. Since the elements arrive in a random order, the expected optimum on the first $n'$ elements
is at least $1/2$ of the actual optimum. Hence, in expectation we obtain at least
$$ p_i \frac{OPT}{2 \alpha} \geq \frac{\eps}{1+\eps} \cdot \frac{1}{(1+i)^{1+\eps}} \cdot \frac{OPT}{2 \alpha}
 \geq \frac{\eps}{4 \alpha} \frac{1}{(1+\log n)^{1+\eps}} OPT.$$
\end{proof}

Therefore, if we run  the $O(\log{r})$ approximation of Babaioff et al. \cite{BIK07} as a black box we achieve an $O(\frac{1}{\eps}\log^{1+\eps}{n}\log{r})$ for the RO-AA-MN model:
\begin{corollary}
For any fixed $\eps > 0$, there is an $O(\frac{1}{\eps}\log^{1+\eps}{n}\log{r})$-approximation for the matroid secretary problem for a matroid $\cM$ of rank $r$ on $n$ elements, with no information given in advance (the RO-AA-MU model). In particular, assuming that $\cM$ is a partition matroid matroid of rank 1, we obtain an $O(\log^{1+\eps}{n} / \eps)$ approximation for the classical secretary problem, with no information given in advance(the CU model).
\end{corollary}

We shall see in Section~\ref{sec:RA-CU-exp} that even in the case of $r=1$ (expectation-maximizing classical secretary problem) where
$n$ is chosen adversarially from $\{1,\ldots,N\}$, we cannot achieve a factor better than $O(\log N / \log \log N)$.

%
%
%

%
%
%
\section{Classical secretary with unknown $n$}
\label{sec:ra-cn}
In this section, we consider a variant of the classical secretary problem where we want to select exactly one element
(i.e. in matroid language, we consider a uniform matroid of rank $1$).
However, here we assume that the total number of elements $n$ (which is crucial in the classical $1/e$-competitive
algorithm) is not known in advance - it is chosen by an adversary who can effectively terminate the input at any point.
We consider the worst case, i.e. we want to achieve a certain probability of success
regardless of when the input is terminated. We show that there is no algorithm achieving a constant probability
of success in this case. However, we can achieve logarithmic guarantees and also prove closely matching lower bounds (see subsection \ref{subsec:classicalsecretary}).

In subsection \ref{sec:RA-CU-exp} we show that even if we want to maximize the expected weight of the selected element,
and $n$ is known to be upper bounded by $N$, still no algorithm can achieve a better than $\Omega(\log{N}/\log\log{N})$ approximation factor in expectation. Consequently, we obtain that no algorithm can achieve an approximation factor better than $\Omega(\log N/\log\log N)$ in the AO-RA-MN model.

\subsection{Known upper bound on $n$}
\label{subsec:classicalsecretary}
First, let us consider the following scenario: an upper bound $N$ is given such that the actual number of
elements on the input is guaranteed to be $n \in \{1,2,\ldots,N\}$.
The adversary can choose any $n$ in this range and we do not learn $n$ until we process the $n$-th element.
(e.g., we are interviewing candidates
for a position and we know that the total number of candidates is certainly not going to be more than $1000$. But, we might run out of candidates at any point.)
The goal is to select the highest-ranking element with a certain probability.
Assuming the {\em comparison model} (i.e., where only the relative ranks of elements are known to the algorithm),
we show that there is no algorithm achieving a constant probability
of success in this case.

\begin{theorem}
\label{thm:RA-CN}
Given that the number of elements is chosen by an adversary in $\{1,\ldots,N\}$
and $N$ is given in advance,
there is a randomized algorithm which selects the best element out of the first $n$
with probability at least $1/(H_{N-1}+1)$.

On the other hand, there is no algorithm in this setting which returns the best element with probability
more than $1/H_N$. Here, $H_N = \sum_{i=1}^{n} \frac{1}{i}$ is the $N$-th harmonic number.

\end{theorem}



Our proof is based on the method of Buchbinder et al.~\cite{BJS10} which bounds the optimal achievable
probability by a linear program. In fact the optimum of the linear program is {\em exactly} the optimal
probability that can be achieved.

\begin{lemma}
\label{lem:RA-CN-LP}
Given the classical secretary problem where the number of elements is chosen by an adversary
from $\{1,2,\ldots,N\}$ and $N$ is known in advance, the best possible probability with which
an algorithm can find the optimal element is given by
\begin{eqnarray}
\nonumber
 \max & \alpha: \\
\label{eq:obj-bound}
 \forall n \leq N; & \frac{1}{n} \sum_{i=1}^{n} i p_i \geq \alpha, \\
\label{eq:p-sum}
 \forall i \leq N; & \sum_{j=1}^{i-1} p_j + i p_i \leq 1, \\
\nonumber
 \forall i \leq N;  & p_i \geq 0.
\end{eqnarray}
\end{lemma}

The only difference between this LP and the one in \cite{BJS10} is that we have multiple constraints
 (\ref{eq:obj-bound}) instead of what is the objective function in \cite{BJS10}.
We use essentially the same proof to argue that this LP captures {\em exactly}
the optimal probability of success $\alpha$ that an algorithm can achieve.
We give the proof for completeness; understanding the validity of this LP will be also useful for us later.

\begin{proof}
Consider any (randomized) algorithm which finds the best element with probability at least $\alpha$,
for every possible number of incoming elements $n \in \{1,\ldots,N\}$.
It is convenient to assume that the algorithm
never learns $n$ and possibly continues running beyond the first $n$ elements (in which case it has failed).
Let us define
$$ p_i = \P{\mbox{algorithm skips the first } i-1 \mbox{ candidates and chooses candidate }i}.$$
The probability here is over both the randomness on the input and the randomness of the algorithm itself.
Recall that the actual number of candidates $n$ is not known beforehand. All that the algorithm knows at time $i$
are the relative ranks of the first $i$ candidates, which are also independent of $n$.
So the probabilities $p_i$ cannot depend on $n$.

Note that these are probabilities of disjoint events, so we have $\sum_{i=1}^{n} p_i \leq 1$.
The LP actually contains stronger inequalities (\ref{eq:p-sum}).
The reason why these inequalities are valid is as follows: We can assume w.l.o.g.~that the algorithm
never selects an element which is not the best so far. (Any algorithm can be converted to this form
and perform at least as well.) The probability (over random permutations of the input) that
the $i$-th candidate is the best so far is $1/i$. Therefore,
$$ \P{\mbox{algorithm skips the first } i-1 \mbox{ and chooses  } i \mid
 \mbox{candidate }i \mbox{ is the best out of the first }i} = $$
$$ = \frac{\P{\mbox{algorithm skips the first } i-1 \mbox{ and chooses candidate } i}}
  {\P{\mbox{candidate }i \mbox{ is the best out of the first }i}}
 = i p_i.$$
On the other hand, the probability that the algorithm skips the first $i-1$ elements is
$1 - \sum_{j=1}^{i-1} p_j$. This event is independent of whether the $i$-th element is
the best among the first $i$,
because all the algorithm learns about the first $i-1$ elements are their relative ranks.
This proves the constraint (\ref{eq:p-sum}):
$$ 1 - \sum_{j=1}^{i-1} p_j = \P{\mbox{algorithm skips the first } i-1 
 \mid \mbox{candidate }i \mbox{ is the best among the first }i} \leq i p_i.$$

The probability that the $i$-th candidate is the actual best candidate among the first $n$ is $1/n$.
Conditioned on this event, candidate $i$ is also the best among the first $i$ candidates (and that is the only information
available to the algorithm at that moment), so the algorithm selects candidate $i$ with conditional probability exactly $i p_i$.
The total probability that the algorithm selects the best candidate out of the first $n$ elements is
$$ \P{\mbox{success}} =  \sum_{i=1}^{n} \P{\mbox{element }i\mbox{ is optimal} \ \& \
 \mbox{algorithm selects }i} = \sum_{i=1}^{n} \frac{1}{n} \cdot i p_i.$$
We assume that the algorithm achieves success probability $\alpha$ for any number of candidates
$n \in \{1,\ldots,N\}$ chosen by an adversary. This proves the constraint (\ref{eq:obj-bound}).

Conversely, given a feasible solution to this LP, an algorithm can proceed as follows (see \cite{BJS10}):
If it comes to the $i$-th element and this is the best element so far, take it with probability
$i p_i / (1 - \sum_{j=1}^{i-1} p_j)$ (which is at most $1$ by (\ref{eq:p-sum}).
It can be verified by induction that the probability
of skipping the first $i-1$ elements and finding that element $i$ is the best so far
is $(1 - \sum_{j=1}^{i} p_j) / i$, and hence the total probability of taking element $i$
is exactly $p_i$. Conditioned on element $i$ being the actual optimum (which happens
with probability $1/n$), we take it with probability $i p_i$.
By (\ref{eq:obj-bound}), the success probability is at least $\alpha$ for any input length $n$.
\end{proof}

For a given $N$, an algorithm can explicitly solve the LP given by Lemma~\ref{lem:RA-CN-LP}
and thus achieve the optimal probability.
Theorem~\ref{thm:RA-CN} can be proved by estimating the value of this LP.

\begin{proofof}{Theorem \ref{thm:RA-CN}}
First, we show a feasible solution with $\alpha = \frac{1}{H_{N-1}+1}$.
We define $p_i = \frac{1}{i (H_{N-1}+1)}$ for each $i=1,\ldots,N$. This induces an algorithm as described
above: if it comes to the $i$-the element and it is the best so far, we take it with probability
$$ \frac{i p_i}{1 - \sum_{j=1}^{i-1} p_j} = \frac{1}{H_{N-1} + 1 - H_{i-1}}.$$
By Lemma~\ref{lem:RA-CN-LP}, it is sufficient to verify that $(p_i,\alpha)$ is a feasible solution:
$$ \frac{1}{n} \sum_{i=1}^{n} i p_i = \frac{1}{H_{N-1}+1} = \alpha $$
implies (\ref{eq:obj-bound}), and
$$ i p _i + \sum_{j=1}^{i-1} p_j = \frac{1}{H_{N-1}+1} (1 + \sum_{j=1}^{i-1} \frac{1}{j})
 = \frac{1}{H_{N-1}+1} (1 + H_{i-1}) \leq 1 $$
implies (\ref{eq:p-sum}). This proves that there is an algorithm with probability of success $1/(H_{N-1}+1)$.

Conversely, we prove that for any feasible solution, we have $\alpha \leq 1/H_N$.
For this, we in fact consider a weaker LP:
\begin{eqnarray}
\nonumber
 \max & \alpha: \\
\label{eq:obj-bound2}
 \forall n \leq N; & \frac{1}{n} \sum_{i=1}^{n} i p_i \geq \alpha, \\
\label{eq:p-sum2}
 & \sum_{i=1}^{N} p_i \leq 1, \\
 \nonumber
 \forall i \leq N;  & p_i \geq 0.
\end{eqnarray}

Obviously, any feasible solution to (\ref{eq:obj-bound}-\ref{eq:p-sum}) is also feasible for
(\ref{eq:obj-bound2}-\ref{eq:p-sum2}).
Fixing $\alpha$, consider a feasible solution to (\ref{eq:obj-bound2}-\ref{eq:p-sum2})
which minimizes $\sum_{i=1}^{N} p_i$.
We claim that $i p_i \geq \alpha$ for each $i$. If not, take the first index $j$ such that $j p_j < \alpha$.
By (\ref{eq:obj-bound2}) for $n=j$, there must be a smaller index $j'<j$ such that $j' p_{j'} > \alpha$.
Then we can decrease $p_{j'}$ by $\delta/j'$ and increase $p_j$ by $\delta/j$ for some small $\delta>0$,
so that $j' p_{j'} + j p_j$ is preserved.
We can make sure that no inequality (\ref{eq:obj-bound2}) is violated, because the left-hand side
is preserved for all $n \geq j$, and the inequality was not tight for $j' \leq n < j$.
On the other hand, $\sum_{i=1}^{N} p_i$ decreases by $\delta/j' - \delta/j$. This is a contradiction.

Therefore, we have $p_i \geq \alpha/i$ for all $i$.
By summing up over all $i$ and using $\sum_{i=1}^{N} p_i \leq 1$,
we get
$$ 1 \geq \sum_{i=1}^{N} p_i \geq \alpha \sum_{i=1}^{N} \frac{1}{i} = \alpha H_N.$$
\end{proofof}

\subsection{Maximizing the expected weight}
\label{sec:RA-CU-exp}
A slightly different model arises when elements arrive with (random) weights and we want
to maximize the expected weight of the selected element. This model is somewhat easier for an algorithm;
any algorithm that selects the best element with probability at least $\alpha$ certainly achieves an $\alpha$-approximation
in this model, but not the other way around.
Given an upper bound $N$
on the number of elements (and under a more stringent assumption that weights are chosen i.i.d. from a known distribution), by a careful choice of a probability distribution for the weights, we prove that still no algorithm can achieve an approximation factor better than an $\Omega(\log N / \log \log N)$-approximation.

\begin{theorem}
\label{thm:max-expect-unknown-n}
For the classical secretary problem with random nonnegative weights drawn i.i.d.~from a known distribution and
the number of candidates chosen adversarially in the range $\{1,\ldots,N\}$, no algorithm achieves
a better than $ \frac{\log N}{32 \log \log N}$-approximation in expectation.
\end{theorem}
The hard examples are constructed based on a particular exponentially distributed probability distribution. Similar constructions have been used in related contexts \cite{HKS07,FHKMS10}.
\begin{proof}
We define a probability distribution over weights as follows. For a parameter $\gamma \in (0,\frac13)$
(possibly depending on $N$), let the weight of each element be (independently)
\begin{itemize}
\item $w_j = 2^{\gamma j}$ with probability $1/2^j$, for each $j \geq 1$.
\end{itemize}
Note that although the weights are unbounded, the expected weight of each element is finite.

Consider blocks of elements where the $i$-th block $B_i$ has size $2^i$. The adversary will choose arbitrarily
a number of blocks $\ell \leq \log N$, and a stopping point $n = \sum_{i=1}^{\ell} 2^i = 2^{\ell+1}-1$.
Note that given $\ell$, the expected optimum is
$$ OPT_\ell \geq \sum_{j=1}^{\infty} w_j \P{w_j \mbox{ is the largest weight among } 2^{\ell} \mbox{ elements}}.$$
The probability that no weight larger than $w_j$ appears among $2^{\ell}$ elements is
$(1-1/2^j)^{2^{\ell}}$. So,
\begin{eqnarray*}
 \P{w_j \mbox{ is the largest weight among } 2^{\ell} \mbox{ elements}}
 & = & (1-1/2^j)^{2^{\ell}} - (1-1/2^{j-1})^{2^{\ell}} \\
 & = & \Theta(\min \{ 2^{\ell-j}, 1 \}).
\end{eqnarray*}
Therefore, since $w_j = 2^{\gamma j}$ and $\gamma \in (0,\frac13)$, the expected contribution from elements
of weight $w_j$ is roughly $2^{\gamma j} \min \{ 2^{\ell-j}, 1\}$, which is maximized for $j = \ell$.
(Note also that the distribution decays exponentially both for $j > \ell$ and $j < \ell$.)
So the largest contribution comes from elements of weight roughly $w_\ell$. We can estimate:
$$ OPT_\ell \geq w_\ell \P{w_\ell \mbox{ appears among }2^\ell \mbox{ elements}}
 = 2^{\gamma \ell} (1 - (1 - 1/2^\ell)^{2^\ell}) \geq (1-1/e) 2^{\gamma \ell}.$$

Now consider any algorithm (which does not know $\ell$ beforehand). Let $p_i$ denote the probability that
the algorithm skips the first $i-1$ blocks and then chooses some element in block $B_i$. Note that this event
might be correlated with the random weights that appear in blocks $B_1,\ldots,B_i$.
However, we have a bound on the probability that weight $w_j$ appears in block $B_i$:
$$ \P{w_j \mbox{ appears in block }B_i} = 1 - (1 - 1/2^j)^{2^i} \leq \min \{ 2^{i-j}, 1 \}.$$
Let $p_{ij}$ denote the probability that the algorithm gets an element of weight $w_j$ from block $B_i$.
By the above we have $p_{ij} \leq \min \{2^{i-j}, 1\}$. Also, by definition of the probabilities, $\sum_{j=1}^{\infty} p_{ij} = p_i$.
Given $p_{ij}$, the expected weight that the algorithm obtains from block $B_i$ is
$ \E{\mbox{profit from }B_i} = \sum_{j=1}^{\infty} w_j p_{ij} $
and the total profit over the first $\ell$ blocks is $\sum_{i=1}^{\ell} \sum_{j=1}^{\infty} w_j p_{ij}$.
Thus the expected profit of any algorithm can be bounded by the following LP.
\begin{eqnarray*}
\max & \alpha: \\
\forall \ell \leq \log N; &\sum_{i=1}^{\ell} \sum_{j=1}^{\infty} w_j p_{ij} \geq \alpha OPT_\ell; \\
\forall i,j;& p_{ij} \leq \min \{ 2^{i-j}, 1 \}; \\
\forall i;& \sum_{j=1}^{\infty} p_{ij} = p_i; \\
&\sum_{i=1}^{\ell} p_i \leq 1; \\
&p_i \geq 0.
\end{eqnarray*}

We estimate the value of this LP as follows.
Subject to the condition $\sum_{j=1}^{\infty} p_{ij} = p_i$, the quantity $\sum_{j=1}^{\infty} w_j p_{ij}$
will be maximized if we make $p_{ij}$ for large $j$ as large as possible.
However, note that $2^{\gamma j} p_{ij} \leq 2^{\gamma j} 2^{i-j}$, so the tail for $j \rightarrow \infty$
decays exponentially and we might as well
concentrate only on the first term. Assuming that $p_i = 2^{i-k}$,
the best choice is to set $p_{ij} = 0$ for $j \leq k$
and $p_{ij} = 2^{i-j}$ for all $j \geq k+1$, which gives
$$ \sum_{j=1}^{\infty} w_j p_{ij} \leq \sum_{j=k+1}^{\infty} 2^{\gamma j} 2^{i-j} \leq \frac{1}{1-2^{\gamma-1}} 2^{(\gamma-1)(k+1) + i}
 \leq 2 \cdot (2^{i-k})^{1-\gamma} 2^{\gamma i} $$
where we used $\gamma \in (0,\frac13)$. Note that for any value of $p_i$, we can apply this argument to the power of $2$ nearest to $p_i$;
hence,
$$ \sum_{j=1}^{\infty} w_j p_{ij} \leq 4 \cdot p_i^{1-\gamma} 2^{\gamma i}.$$
Now suppose the adversary stops the game after $\ell$ blocks. The expected optimum is
$OPT_\ell \geq (1-1/e) 2^{\gamma \ell}$
(see above), while the algorithm gets
$$ \sum_{i=1}^{\ell} \sum_{j=1}^{\infty} w_j p_{ij} \leq 4 \cdot \sum_{i=1}^{\ell} p_i^{1-\gamma} 2^{\gamma i}.$$
This should be at least $\alpha OPT_\ell \geq \alpha (1-1/e) 2^{\gamma \ell}$; therefore, we get
$$ \sum_{i=1}^{\ell} p_i^{1-\gamma} 2^{\gamma (i-\ell)} \geq \frac14 (1-1/e) \alpha \geq \frac18 \alpha.$$
We sum up these inequalities for $\ell=1,\ldots,\log N$:
$$ \sum_{\ell=1}^{\log N} \sum_{i=1}^{\ell} p_i^{1-\gamma} 2^{\gamma (i-\ell)}
 = \sum_{i=1}^{\log N} p_i^{1-\gamma} \sum_{\ell=i}^{\log N} 2^{\gamma (i-\ell)} \geq \frac18 \alpha \log N.$$
The sum $\sum_{\ell=i}^{\log N} 2^{\gamma (i-\ell)}$ is bounded by $\sum_{\ell=i}^{\infty} 2^{\gamma (i-\ell)} = \frac{1}{1-2^{-\gamma}}
 \leq \frac{2}{\gamma}$. Therefore, we get
$$ \alpha \leq \frac{16}{\gamma \log N} \sum_{i=1}^{\log N} p_i^{1-\gamma}.$$
Given that $\sum_{i=1}^{\log N} p_i = 1$ and the function $x^{1-\gamma}$ is concave,
the best value of $\alpha$ can be achieved if we set $p_i = 1/\log N$ for all $i$. Then, we have
$$ \alpha \leq \frac{16}{\gamma (\log N)^{1-\gamma}}.$$
Finally, we set $\gamma = 1 / \log \log N$ which gives
$$ \alpha \leq \frac{32 \log \log N}{\log N}.$$
\end{proof}

Consequently, we obtain that no algorithm can achieve an approximation factor better than $\Omega(\log N/\log\log N)$ in the AO-RA-MN model.

\begin{corollary}
\label{cor:hardness_aoramn}
For the matroid secretary problem in the AO-RA-MN (and AO-RA-MU, RO-AA-MU) models,
no algorithm can achieve a better than $\Omega(\frac{\log{N}}{\log\log{N}})$-approximation in expectation.
\end{corollary}
\begin{proof}
It is not hard to convert the example of Theorem \ref{thm:max-expect-unknown-n} into a hard example for the AO-RA-MN model. It suffices to let $\cM$ to be a partition matroid of rank 1, and let the first $n$ elements of the inputs to be non-loop while the rest of the input contains only loops. Since the algorithm does not know $n$ in advance (it only knows $n\leq N$), it essentially has to choose one of the first $n$ elements without knowing $n$, which is a secretary problem where the number of candidates is chosen adversarially in the range $\{1,\ldots,N\}$. Therefore, no algorithm can achieve an approximation factor better than $\Omega(\log{N}/\log\log{N})$ (the same is also true for the AO-RA-MU, RO-AA-MU model,
where nothing is known about $n$ in advance).
\end{proof}

\section{Conclusion and open questions}

We presented a number of positive and negative results for variants of the matroid secretary problem.
The main open question is if there is a constant-factor approximation in the RO-AA-MN model, where weights are
assigned to elements adversarially and the input ordering of elements is random. An easier question might be whether
this is possible in the RO-RA-MN model where both the input order and weight assignment are random,
but only the total number of elements $n$ is known in advance (as opposed to the full matroid structure, as in \cite{Soto11}).
Note that under an adversarial assignment of weights, knowing the matroid beforehand (RO-AA-MK) does not seem to be
easier than the RO-AA-MN model; the true input could be embedded in a much larger matroid with most weights set to zero.
A similar question arises for the AO-RA-MN model: whether it is possible to improve the $O(\log n\log r)$ factor, thus closing the gap with the lower-bound of $\Omega(\log n/\log\log n)$.


\begin{thebibliography}{99}

\bibitem{ABT82} A. R. Abdel-Hamid, J.A. Bather and G. B. Trustrum.
The secretary problem with an unknown number of candidates.
{\em J. Appl. Prob.} 19, 619--630, 1982.

\bibitem{BDGIT09} M. Babaioff, M. Dinitz, A. Gupta, N. Immorlica and K. Talwar.
Secretary problems: weights and discounts.
In {\em SODA 2009}, 1245--1254.

\bibitem{BIK07} M. Babaioff, N. Immorlica and R. Kleinberg.
Matroids, secretary problems, and online mechanisms.
In {\em SODA 2007}, 434--443.

\bibitem{BHZ10} M. H. Bateni, M. T. Hajiaghayi and M. Zadimoghaddam.
Submodular secretary problem and extensions.
In {\em APPROX 2010}, 39--52.

\bibitem{BS09} P. Borosan and M. Shabbir.
A survey of secretary problem and its extensions.
unpublished manuscript, 2009, available at
\url{http://paul.rutgers.edu/~mudassir/Secretary/paper.pdf}.

\bibitem{BJS10} N. Buchbinder, K. Jain and  M. Singh.
Secretary problems via linear programming.
In {\em IPCO 2010}, 163--176.

\bibitem{CL12} S. Chakraborty and O. Lachish.
Improved competitive ratio for the matroid secretary problem.
To appear in {\em SODA 2012}.


\bibitem{DP08} N. B. Dimitrov and C. G. Plaxton.
Competitive weighted matching in transversal matroids.
In {\em ICALP 2008}, 397--408.

\bibitem{Dynkin63} E. B. Dynkin.
The optimum choice of the instant for stopping a markov process.
{\em Soviet Mathematics, Doklady} 4, 1963.

\bibitem{FHKMS10} J. Feldman, M. Henzinger, N. Korula, V. S. Mirrokni and C. Stein.
Online stochastic packing applied to display ad allocation.
In {\em ESA 2010}, 182--194.

\bibitem{Ferguson89} T.S. Ferguson.
Who solved the secretary problem?
{\em Statistical Science}, 4:3, 282--289, 1989.

\bibitem{Gardner60} M. Gardner.
Mathematical Games column,
{\em Scientific American}, February 1960.

\bibitem{GM66} J. Gilbert and F. Mosteller.
Recognizing the maximum of a sequence.
{\em J. Amer. Statist. Assoc.} 61:35--73, 1966.

\bibitem{GRST10} A. Gupta, A. Roth, G. Schoenebeck and K. Talwar.
Constrained non-monotone submodular maximization: Offline and secretary algorithms.
In {\em WINE 2010}, 246--257.

\bibitem{HKP04} M. T. Hajiaghayi, R. Kleinberg and D. Parkes.
Adaptive limited-supply online auctions.
In {\em EC 2004}, 71--80.

\bibitem{HKS07} M. T. Hajiaghayi,  R. Kleinberg and T. Sandholm.
Automated online mechanism design and prophet inequalities.
In {\em International Conference on Artificial Intelligence 2007}, 58--65, 2007.

\bibitem{IW11} S. Im and Y. Wang.
Secretary problems: Laminar matroid and interval scheduling.
In {\em SODA 2011}, 1265-1274.

\bibitem{Lindley61} D. V. Lindley.
Dynamic programming and decision theory.
{\em Journal of the Royal Statistical Society. Series C (Applied Statistics)}, 10:1, 39--51, 1961.

\bibitem{Kleinberg05} R. Kleinberg.
A multiple-choice secretary algorithm with applications to online auctions.
In {\em SODA 2005}, 630--631.

\bibitem{KP09}N. Korula and M. P\'al.
Algorithms for secretary problems on graphs and hypergraphs.
In {\em ICALP 2009}, 508--520.

\bibitem{Soto11} J. A. Soto.
Matroid secretary problem in the random assignment model.
In {\em SODA 2011}, 1275-1284.

\bibitem{SotoPhD} J. Soto.
Contributions on secretary problems, independent sets of rectangles and related problems.
PhD Thesis, Department of Mathematics, Massachusetts Institute of Technology, 2011.


\bibitem{Stewart80} T. J. Stewart.
The secretary problem with an unknown number of options.
{\em Operations Research} 29:1, 130--145, 1981.



\end{thebibliography}
\end{document}